\documentclass[article,oneside,12pt,a4paper]{memoir}

\usepackage[utf8]{inputenc}
\usepackage[T1]{fontenc}
\usepackage[shrink=10,stretch=10]{microtype}
\usepackage{amsmath,amssymb,amsthm}
\usepackage[british]{babel}
\usepackage[svgnames]{xcolor}
\usepackage{graphicx}
\usepackage{booktabs}
\usepackage{calc}
\usepackage[hidelinks]{hyperref}
\usepackage{tikz}
\usetikzlibrary{patterns}

\setlrmarginsandblock{3.7cm}{3.7cm}{*}
\setulmarginsandblock{3cm}{3cm}{*}
\checkandfixthelayout

\makeatletter
\def\@endtheorem{\endtrivlist}
\makeatother
\newtheorem{theorem}{Theorem}
\newtheorem{lemma}[theorem]{Lemma}
\newtheorem{proposition}[theorem]{Proposition}
\newtheorem{corollary}[theorem]{Corollary}
\newtheorem{definition}[theorem]{Definition}

\newtheorem{example}{Example}

\hypersetup{
    pdftitle={Steane-Enlargement of Quantum Codes from the Hermitian Curve},
    pdfauthor={René Bødker Christensen and Olav Geil}
}

\captionnamefont{\footnotesize\bfseries}	
\captiontitlefont{\footnotesize}					
\captiondelim{. }
\setsecheadstyle{\bfseries\large}

\counterwithout{section}{chapter}
\setsubsecheadstyle{\bfseries}

\renewcommand{\vec}[1]{\mathbf{#1}}
\newcommand{\ph}{\phantom{0}}
\newcommand{\imprMark}{\ensuremath{\ast}}
\newcommand{\phMark}{\phantom{\imprMark}}
\newcommand{\sep}{\unskip,\ }
\newcommand{\orcidID}[1]{
        \unskip\hspace{0.1em}\href{https://orcid.org/#1}{\raisebox{-0.12\height}{\includegraphics[height=0.8em]{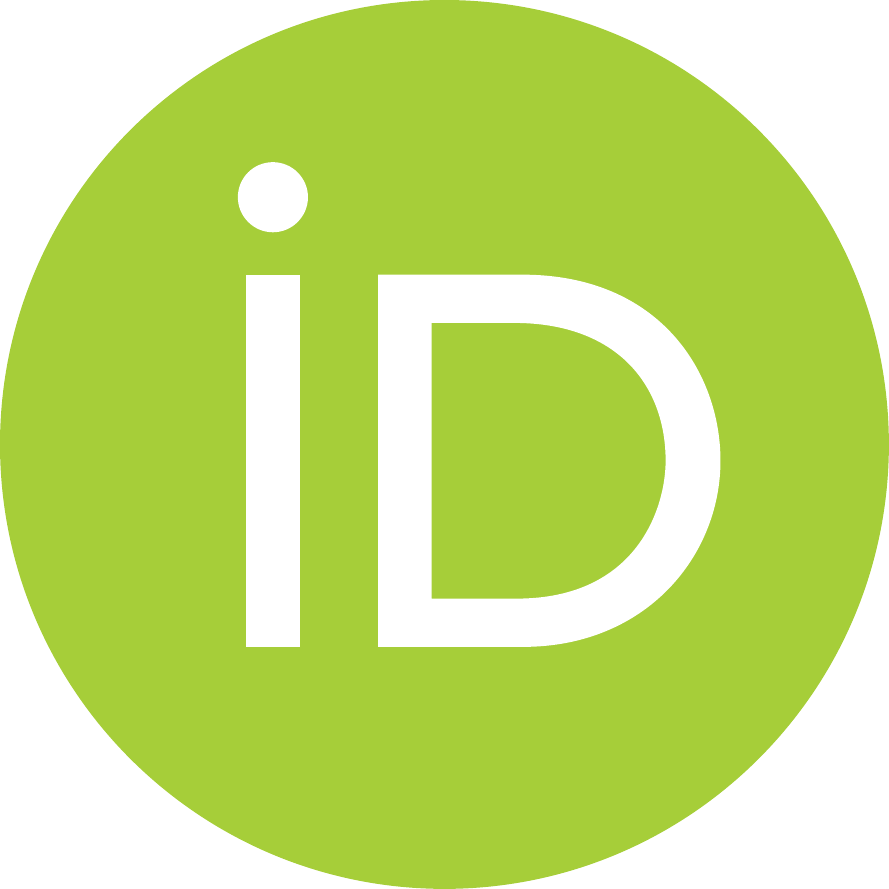}}}
    }

\begin{document}

\begin{center}
  {\Large\fontfamily{ppl}\fontseries{b}\selectfont Steane-Enlargement of Quantum Codes\\\smallskip from the Hermitian Curve}\\[1.5em]
  {René Bødker Christensen\orcidID{0000-0002-9209-3739}and Olav Geil\orcidID{0000-0002-9666-3399}}
  
  \bigskip{Department of Mathematical Sciences,
      Aalborg University, Denmark.}
  
  {\texttt{\{rene,olav\}@math.aau.dk}}

  \vglue 2em
  {\small\textbf{Abstract}}\\[0.5em]
  \begin{minipage}{0.9\linewidth}
    \rule{\linewidth}{0.5pt}
    \footnotesize In this paper, we study the construction of quantum codes by applying Steane-enlargement to codes from the Hermitian curve. We cover Steane-enlargement of both usual one-point Hermitian codes and of order bound improved Hermitian codes. In particular, the paper contains two constructions of quantum codes whose parameters are described by explicit formulae, and we show that these codes compare favourably to existing, comparable constructions in the literature.

    \bigskip\emph{Keywords:} Algebraic geometric code \sep Quantum code \sep Steane-enlargement \sep Hermitian curve

    \medskip\emph{2000 MSC:} 94B27 \sep 81Q99\\[-5pt]
    \rule{\linewidth}{0.5pt}
  \end{minipage}
\end{center}

\section{Introduction}
The prospect of quantum computers potentially surpassing the computational abilities of classical computers has spawned much interest in studying and building large-scale quantum computers. Since such quantum systems would be very susceptible to disturbances from the environment and to imperfections in the quantum gates acting on the system, the implementation of a working quantum computer requires some form of error-correction. This has led to the study of quantum error-correcting codes, and although such codes are conceptually similar to their classical brethren, their construction call for different techniques. Nevertheless, results have been found that link classical codes to quantum ones, suggesting that good quantum codes may be found by considering good classical codes.

A well-known class of algebraic geometric codes is the one-point codes from the Hermitian curve.
For these one-point Hermitian codes, one of the simplest bounds on the minimal distance is the Goppa bound. For codes of sufficiently large dimension, however, the Goppa bound does not give the true minimal distance, and the order bound for dual codes \cite{handbook,DP10} and for primary codes \cite{GMRT11,AG08,NormTrace} give more information on the minimal distance of the codes. These improved bounds also give rise to a family of improved codes with designed minimal distances, and we shall refer to such codes as \emph{order bound improved} codes.

The construction of quantum codes from one-point Hermitian codes has already been considered in \cite{SK06}, and from order bound improved Hermitian codes in \cite{qcHermit}. Neither of these works, however, explore the potential benefit from applying Steane-enlargement to the codes under consideration. Thus, this paper will address this question, and describe the quantum codes that can be obtained in this manner.

The work is structured as follows. Section~\ref{sec:preliminaries} contains the preliminary theory on quantum codes and order bound improved Hermitian codes that will be necessary in the subsequent sections. Afterwards, Section~\ref{sec:steaneHermitian} covers the results of applying Steane-enlargement to one-point Hermitian codes and order bound improved Hermitian codes. The parameters of the resulting codes are then compared to those of codes already in the literature in Section~\ref{sec:comp-with-exist}.
\section{Preliminaries}\label{sec:preliminaries}
In this section, we shall reiterate the necessary definitions and results regarding both quantum error-correcting codes and order bound improved Hermitian codes. For both of these, we will be relying on nested pairs of classical codes, and on the relative distance of such pairs. Thus, recall that for classical, linear codes $\mathcal{C}_2\subsetneq\mathcal{C}_1$, we define the relative distance of the pair as
\begin{equation*}
  d(\mathcal{C}_1,\mathcal{C}_2)=\min\{w_H(\vec{c})\mid\vec{c}\in\mathcal{C}_1\setminus\mathcal{C}_2\},
\end{equation*}
where $w_H$ denotes the Hamming weight.

\subsection{Quantum codes}
A $k$-dimensional quantum code of length $n$ over $\mathbb{F}_q$ is a $q^k$-dimensional subspace of the Hilbert space $\mathbb{C}^{q^n}$. This space is subject to phase-shift errors, bit-flip errors, and combinations thereof. For a quantum code, we define its two minimal distances $d_z$ and $d_x$ as the maximal integers such that the code allows simultaneous detection of any $d_z-1$ phase-shift errors and any $d_x-1$ bit-flip errors. When such a code has length $n$ and dimension $k$, we refer to it as an $[[n,k,d_z/d_x]]_q$-quantum code.

The literature contains many works based on the assumption that it is not necessary to distinguish between the two types of errors. Thus, the quantum code is only associated with a single minimal distance. That is, we say that its minimal distance is $d=\min\{d_z,d_x\}$, and the notation for the parameters is presented slightly more compactly as $[[n,k,d]]_q$. In this case, we refer to the quantum code as being \emph{symmetric}, and in the previous case we refer to it as being \emph{asymmetric}.

One of the commonly used constructions of quantum codes was provided by Calderbank, Shore, and Steane \cite{CalderbankShor,Steane96}, and relies on a self-orthogonal, classical error-correcting code in order to obtain a quantum stabilizer code. It was later shown that the self-orthogonal code can be replaced by a pair of nested codes, giving asymmetric quantum codes. This generalized CSS-construction is captured in the following theorem found in \cite{SKR09}.
\begin{theorem}\label{thm:css}
  Given $\mathbb{F}_q$-linear codes $\mathcal{C}_2\subsetneq\mathcal{C}_1$ of length $n$ and codimension $\ell$, the CSS-construc\-tion ensures the existence of an asymmetric quantum code with parameters
  \begin{equation*}
    [[n,\ell,d_z/d_x]]_q
  \end{equation*}
  where $d_z=d(\mathcal{C}_1,\mathcal{C}_2)$ and $d_x=d(\mathcal{C}_2^\perp,\mathcal{C}_1^\perp)$.
\end{theorem}
\begin{corollary}\label{cor:steaneSelfOrth}
  If the $[n,k,d]$ linear code $\mathcal{C}\subseteq\mathbb{F}_q^n$ is self-orthogonal, then a
  \begin{equation*}
    [[n,2k-n,d]]_q
  \end{equation*}
  symmetric quantum code exists.
\end{corollary}
When the CSS-construction is applied to a self-orthogonal binary linear code as in Corollary~\ref{cor:steaneSelfOrth}, Steane \cite{Steane99} proposed a procedure whereby the dimension of the resulting quantum code may be increased. In the best case, this can be done with little or no decrease in the minimal distance of the quantum code. This procedure -- eponymously named Steane-enlargement in the literature -- has later been generalized to $q$-ary codes as well \cite{Hamada,qArySteane}. 
\begin{theorem}\label{thm:qArySteane}
  Consider a linear $[n,k]$ code $\mathcal{C}\subseteq\mathbb{F}_q^n$ that contains its Euclidean dual $\mathcal{C}^\perp$. If $\mathcal{C}'$ is an $[n,k']$ code such that $\mathcal{C}\subsetneq\mathcal{C}'$ and $k'\geq k+2$, then an
  \begin{equation*}
    \Big[\Big[n,k+k'-n,\geq \min\big\{d,\big\lceil(1+\textstyle\frac{1}{q})d'\big\rceil\big\}\Big]\Big]_q
  \end{equation*}
  quantum code exists with $d=d(\mathcal{C},\mathcal{C}'^\perp)$ and $d'=d(\mathcal{C}',\mathcal{C}'^\perp)$.
\end{theorem}
When presenting the parameters of a Steane-enlarged code in propositions of this paper, we will often state the dimension in the form $2k-n+(k'-k)$. In this way, we highlight the dimension increase since $2k-n$ is the dimension of the non-enlarged quantum code.

\subsection{Order bound improved Hermitian codes}
\label{sec:order-bound-improved}
Let $P_1,P_2,\ldots,P_{q^3},Q$ be the rational places of the Hermitian function field over $\mathbb{F}_{q^2}$, and define the divisor $D=P_1+P_2+\cdots+P_{q^3}$. Further, denote by $H(Q)$ the Weierstraß semigroup of $Q$. As in \cite{qcHermit,GMRT11}, we consider a special subset of $H(Q)$, namely
\begin{equation*}
  H^\ast(Q)=\{\lambda\in H(Q)\mid C_{\mathcal{L}}(D,\lambda Q)\neq C_{\mathcal{L}}(D,(\lambda-1)Q)\}.
\end{equation*}
It may be shown that in fact
\begin{equation*}
  H^\ast(Q)=\{iq+j(q+1)\mid 0\leq i<q^2,\; 0\leq j<q\}.
\end{equation*}
Now, fix an element $f_\lambda\in\mathcal{L}(\lambda Q)\setminus\mathcal{L}((\lambda-1)Q)$ for each $\lambda\in H^\ast(Q)$, and define the map $\sigma\colon H^\ast(Q)\rightarrow\mathbb{N}$ given by
\begin{equation*}
  \sigma(iq+j(q+1))=\begin{cases}
    q^3-iq-j(q+1) &\text{if } 0\leq i<q^2-q\\
    (q^2-i)(q-j) &\text{if } q^2-q\leq i<q^2
  \end{cases}.
\end{equation*}
With this map, it was shown in~\cite{qcHermit} as a special case of \cite{NormTrace} that the improved primary code
\begin{equation*}
  \tilde{E}(\delta)=\{(f_\lambda(P_1),f_\lambda(P_2),\ldots,f_\lambda(P_n))\mid\sigma(\lambda)\geq\delta\}
\end{equation*}
has minimal distance exactly $\delta$ whenever $\delta\in \sigma(H^\ast(Q))$. To obtain an improved dual code, let $\mu\colon H^\ast(Q)\rightarrow\mathbb{N}$ be given by $\mu(iq+j(q+1))=\sigma((q^2-1-i)q+(q-1-j)(q+1))$.
Then, as explained in \cite{qcHermit}, the code
\begin{equation*}
  \tilde{C}(\delta)=\big(\{(f_\lambda(P_1),f_\lambda(P_2),\ldots,f_\lambda(P_n))\mid\mu(\lambda)<\delta\}\big)^\perp
\end{equation*}
also has minimal distance exactly $\delta$ whenever $\delta\in\mu(H^\ast(Q))=\sigma(H^\ast(Q))$, and in fact $\tilde{E}(\delta)=\tilde{C}(\delta)$.

For the order bound to produce an improved code, the designed distance must be sufficiently small. Otherwise, the code $\tilde{E}(\delta)$ simply corresponds to one of the usual one-point Hermitian codes. This correspondence is given in the following result from \cite[Cor.\,4]{qcHermit}.
\begin{lemma}\label{lem:improvedEqualUsual}
  For $\delta>q^2-q$ we have $\tilde{E}(\delta)=C_{\mathcal{L}}(D,(q^3-\delta)Q)$, but the code $C_{\mathcal{L}}(D,(q^3-(q^2-q))Q)$ is strictly contained in $\tilde{E}(q^2-q)$.
\end{lemma}
For $\delta\leq q^2-q$, the work \cite{qcHermit} contains a lower bound on the dimension of $\tilde{E}(\delta)$. In Proposition~\ref{prop:hermitDimension} below, we give an explicit formula describing the dimension in this case. This formula relies on the number of (number theoretic) divisors of a certain type, as specified in the following definition.
\begin{definition}
  For $n\in\mathbb{Z}_+$, we let $\tau^{(q)}(n)$ denote the number of divisors $d$ of $n$ such that $0\leq d\leq q$ and $n/d\leq q$.
\end{definition}
From the definition it should be clear that $\tau^{(q)}(n)$ can be computed in $\mathcal{O}(q)$ operations.

\begin{proposition}\label{prop:hermitDimension}
  Let $1\leq\delta\leq q^2$, and write $\delta-1=aq+b$. Then
  \begin{equation*}
    \dim(\tilde{E}(\delta))= 
    q^3-q^2 
    - \frac{a(a-1)}{2} - \min\{a,b\}
    + \sum_{i=\delta}^{q^2}\tau^{(q)}(i).
  \end{equation*}
\end{proposition}
\begin{proof}
  We give the proof by partitioning $H^\ast(Q)$ in three disjoint sets:
  \begin{align*}
    \Lambda_1 &=\{iq+j(q+1)\in H^\ast(Q)\mid i+j< q^2-q,\; 0\leq i<q^2-q,\; 0\leq j<q\}\\
    \Lambda_2 &=\{iq+j(q+1)\in H^\ast(Q)\mid i+j\geq q^2-q,\; 0\leq i< q^2-q,\; 0\leq j<q\}\\
    \Lambda_3 &=\{iq+j(q+1)\in H^\ast(Q)\mid q^2-q\leq i<q^2,\; 0\leq j<q \}. 
  \end{align*}
  We first determine the cardinality of $\Lambda_2$. Considering some $iq+j(q+1)\in\Lambda_2$, and writing $i=q^2-q-k$, there are $q-k$ possible values of $j$.
  There are $q-1$ such integers $k$ since $q^2-2q+1\leq i<q^2-q$ within the set $\Lambda_2$. This implies that 
  \begin{equation*}
    |\Lambda_2|=\sum_{k=1}^{q-1}(q-k)=\frac{q(q-1)}{2}=g,
  \end{equation*}
  where $g$ is the genus of the Hermitian curve.
  From this, it is also seen that $|\Lambda_1|=q^3-q^2-|\Lambda_2|=q^3-q^2-g$.

  All elements $\lambda$ of $\Lambda_1$ satisfy $\sigma(\lambda)=q^3-\lambda$. The largest element $\lambda'$ in $\Lambda_1$ is given by $\lambda'=(q^2-2q)q+(q-1)(q+1)$, which has $\sigma(\lambda')=q^2+1$. Thus, all elements of $\Lambda_1$ have $\sigma(\lambda)\geq\delta$, meaning that $\Lambda_1$ contributes $|\Lambda_1|=q^3-q^2-g$ to the dimension of $\tilde{E}(\delta)$.

  In order to determine the number of elements in $\Lambda_2$ that satisfy $\sigma(\lambda)\geq\delta$, we compute $|\Lambda_2|-|\{\lambda\in\Lambda_2\mid\sigma(\lambda)<\delta\}|$. As was the case for $\Lambda_1$, all elements of $\Lambda_2$ have $\sigma(\lambda)=q^3-\lambda$. From this it follows that
  \begin{equation*}
    \sigma(\Lambda_2)=\{q+1,2q+1,2q+2,3q+1,3q+2,3q+3,4q+1,\ldots, (q-1)q+(q-1)\}.
  \end{equation*}
  Combining this with the assumption that $\delta-1=aq+b$, the number of elements in $\sigma(\Lambda_2)$ smaller than $\delta$ is exactly
  \begin{equation*}
    \sum_{i=1}^{a-1}i+\min\{a,b\}=\frac{a(a-1)}{2}+\min\{a,b\}.
  \end{equation*}
  Because the total number of elements is $|\Lambda_2|=g$, the set $\Lambda_2$ contributes $g-a(a-1)/2-\min\{a,b\}$ to the dimension.
  
  Finally, consider $\sigma(\Lambda_3)=\{\sigma(\lambda)\mid\lambda\in\Lambda_3\}$ as a multiset. We will count (with multiplicity) the number of elements $s\in\sigma(\Lambda_3)$ with $s\geq\delta$. Observe that $\sigma(\lambda)=(q^2-i)(q-j)$ for all the elements $\lambda=iq+j(q+1)\in\Lambda_3$. Hence, $s\in\sigma(\Lambda_3)$, if and only if $s=d\cdot\frac{s}{d}$ where $d\leq q$ and $\frac{s}{d}\leq q$. Since there are $\tau^{(q)}(s)$ such divisors $d$, it follows that the multiplicity of $s$ in $\sigma(\Lambda_3)$ is given by $\tau^{(q)}(s)$. Subsequently, the number of elements satisfying $s\geq\delta$ is
  \begin{equation*}
    \sum_{s=\delta}^{q^2}\tau^{(q)}(s).
  \end{equation*}
  By summing the contribution from each of the sets $\Lambda_1$, $\Lambda_2$, and $\Lambda_3$, we obtain the dimension as claimed.
\end{proof}
We note that dimension formula in Proposition~\ref{prop:hermitDimension} does not provide an efficient method for computing the dimension of the code $\tilde{E}(\delta)$. Since the set $\Lambda_3$ defined in the proof has $q^2$ elements, we can loop over all of these to compute the dimension in $\Theta(q^2)$ operations, whereas the formula in Proposition~\ref{prop:hermitDimension} requires $\mathcal{O}(q^3)$. It does, however, provide an advantage when we are not interested in a dimension, but rather a codimension as will be the case in Section~\ref{sec:steaneHermitian}. Here, we will only need to compute $\tau^{(q)}$ for $m$ values, where $m$ is a small integer; typically $m=1$ or $m=2$.

If only a lower bound for the dimension is needed, Lemma~6 of \cite{qcHermit} implies that the sum in Proposition~\ref{prop:hermitDimension} can be bounded below by $q^2-\lfloor\delta+\delta\ln(q^2/\delta)\rfloor$ for $q\leq\delta <q^2$ and by $q^2-\lfloor\delta+\delta\ln(\delta)\rfloor$ for $\delta<q$.

\section{Steane-enlargement of Hermitian codes}\label{sec:steaneHermitian}
In order to apply Steane-enlargement to the codes defined in Section~\ref{sec:order-bound-improved}, we now determine a necessary and sufficient condition for $\tilde{E}(\delta)$ to be self-orthogonal. While this is possible to do by considering the improved codes directly, it is easier to prove via the self-orthogonality of usual one-point Hermitian codes. The latter is well-known, and the following result was given in \cite{Tiersma}, and can also be found in \cite[Prop.\,8.3.2]{Stichtenoth}.
\begin{proposition}\label{prop:selfOrthUsual}
  The code $C_{\mathcal{L}}(D,(q^3-\delta)Q)$ is self-orthogonal, if and only if
  \begin{equation}\label{eq:hermitianSelfOrthCond}
    \delta\leq \left\lfloor\frac{1}{2}(q^3-q^2+q)\right\rfloor+1.
  \end{equation}
\end{proposition}
\begin{corollary}\label{cor:selfOrthImprov}
  The code $\tilde{E}(\delta)$ is self-orthogonal, if and only if $\delta$ satisfies~\eqref{eq:hermitianSelfOrthCond}.
\end{corollary}
\begin{proof}
  For $\delta>q^2-q$, Lemma~\ref{lem:improvedEqualUsual} ensures that $\tilde{E}(\delta)=C_{\mathcal{}}(D,(q^3-\delta)Q)$, and the result follows from Proposition~\ref{prop:selfOrthUsual}. For smaller values of $\delta$, the result follows from the observation that $\tilde{E}(q^2-q+1)\subsetneq\tilde{E}(\delta)$.
\end{proof}
In Theorem~\ref{thm:qArySteane}, the relative distances of the code pairs are used to determine the distance of the resulting quantum code. In the case of one-point Hermitian codes and order bound improved Hermitian codes, however, the relative and non-relative distances coincide.
To see this, consider codes $\mathcal{C}^\perp\subsetneq\mathcal{C}\subsetneq\mathcal{C}'$. Since $\mathcal{C}$ is self-orthogonal, Proposition~\ref{prop:selfOrthUsual} and Corollary~\ref{cor:selfOrthImprov} ensure that $\mathcal{C}'^\perp\subsetneq C_{\mathcal{L}}(D,(q^3-\delta_{\max})Q)\subseteq \mathcal{C}$ where $\delta_{\max}=\left\lfloor\frac{1}{2}(q^3-q^2+q)\right\rfloor+1$ as in \eqref{eq:hermitianSelfOrthCond}.
In particular, this also implies the inclusions $\mathcal{C}'^\perp\subseteq C_{\mathcal{L}}(D,(q^3-\delta_{\max}-1)Q)$, so that every codeword of $\mathcal{C}'^\perp$ has Hamming weight at least $d(C_{\mathcal{L}}(D,(q^3-\delta_{\max}-1)Q))=\delta_{\max}+1$, which exceeds both $d(\mathcal{C})$ and $d(\mathcal{C}')$. Thus, relative distances satisfy $d(\mathcal{C},\mathcal{C}'^\perp)=d(\mathcal{C})$ and $d(\mathcal{C}',\mathcal{C}'^\perp)=d(\mathcal{C}')$. For this reason, we only need to consider the non-relative distances in the proofs below.

In the following proposition, we explore the Steane-enlargement from Theorem~\ref{thm:qArySteane} applied to the usual one-point Hermitian codes. That is, we show by how much the dimension of the symmetric quantum error correcting code can be increased without decreasing its minimal distance. Before giving the result itself, we state the following lemma, which follows from \cite{YangKumar}.
\begin{lemma}\label{lem:containedWeierstrass}
  If $\lambda\in\mathbb{N}$ satisfies $2g\leq \lambda<q^3$, then $\lambda\in H^\ast(Q)$.
\end{lemma}
\begin{proposition}\label{prop:enlargeUsualHermit}
  Assume that $\delta$ satisfies \eqref{eq:hermitianSelfOrthCond}, and additionally that $\delta\geq q^2+3$. If $k$ denotes the dimension of $C_{\mathcal{L}}(D,(q^3-\delta)Q)$, then there exists a quantum code with parameters
  \begin{equation}\label{eq:qeccUsualHermit}
    \Big[\Big[q^3,2k-q^3+\big\lceil\textstyle\frac{\delta-1}{q^2+1}\big\rceil,\geq\delta\Big]\Big]_{q^2}.
  \end{equation}
\end{proposition}
\begin{proof}
  According to Proposition~\ref{prop:selfOrthUsual}, the code $C_{\mathcal{L}}(D,(q^3-\delta)Q)$ is self-orthogo\-nal. Letting $\delta'=\delta-\lceil(\delta-1)/(q^2+1)\rceil$, it is also seen that $C_{\mathcal{L}}(D,(q^3-\delta)Q)\subseteq C_{\mathcal{L}}(D,(q^3-\delta')Q)$. Lemma~\ref{lem:containedWeierstrass} ensures that the $\lceil(\delta-1)/(q^2+1)\rceil$ integers $\delta-1,\delta-2,\ldots,\delta'$ are all included in $H^\ast(Q)$, meaning that the dimension of $C_{\mathcal{L}}(D,(q^3-\delta')Q)$ is $k+\lceil(\delta-1)/(q^2+1)\rceil\geq k+2$. Thus, we can apply Theorem~\ref{thm:qArySteane} to obtain a quantum code over $\mathbb{F}_{q^2}$ of length and dimension as in~\eqref{eq:qeccUsualHermit}. This code has minimal distance at least $\delta$ since
  \begin{equation*}
    \left(1+\frac{1}{q^2}\right)\delta'> \left(1+\frac{1}{q^2}\right)\left(\delta-\frac{\delta-1}{q^2+1}-1\right)=\delta-1,
  \end{equation*}
  and since Lemma~\ref{lem:improvedEqualUsual} ensures that $d\big(C_{\mathcal{L}}(D,(q^3-\delta)Q)\big)=d(\tilde{E}(\delta))=\delta$.
\end{proof}
We now turn our attention to the order bound improved codes, and begin by considering the case where both codes can be described as improved codes.
\begin{proposition}\label{prop:enlargeImprovedHermit}
  Assume that $\delta\in\sigma(H^\ast(Q))$, and that $2\leq\delta\leq q^2$. Let $k$ denote the dimension of $\tilde{E}(\delta)$, and choose an $m\in\{1,2,\ldots,\delta-1\}$. Write $\delta-1=aq+b$ and $\delta-m-1=a'q+b'$ such that $0\leq b,b'<q$, and define
  \begin{equation}\label{eq:quantumParamsHermit}
    K=\min\{a,b\}-\min\{a',b'\}+\frac{a(a-1)-a'(a'-1)}{2}+\sum_{i=1}^{m}\tau^{(q)}(\delta-i).
  \end{equation}
  If $K\geq 2$, then there exists a $[[q^3,2k-q^3+K,\geq\delta-m+1]]_{q^2}$ quantum code.
\end{proposition}
\begin{proof}
  Consider any $m$ such that $1\leq m<\delta$, and define $\delta'=\delta-m$. By Corollary~\ref{cor:selfOrthImprov}, the code $\tilde{E}(\delta)$ is self-orthogonal. Furthermore, $\tilde{E}(\delta)\subseteq\tilde{E}(\delta')$, and Proposition~\ref{prop:hermitDimension} implies that the dimension difference is $\dim(\tilde{E}(\delta'))-\dim(\tilde{E}(\delta))=K$. Thus, if $K\geq 2$, we can apply Theorem~\ref{thm:qArySteane} to obtain a quantum code, whose dimension is $2k-q^3+K$.
  To determine its minimal distance, we see that
  \begin{equation*}
    \left\lceil\left(1+\frac{1}{q^2}\right)\delta'\right\rceil = \left\lceil(\delta-m)+\frac{\delta-m}{q^2}\right\rceil=\delta-m+1.
  \end{equation*}
  The result follows from the fact that $\min\{\delta,\delta-m+1\}=\delta-m+1$.
\end{proof}
To fully describe the quantum codes that can be constructed using the order bound improved codes, it is also necessary to consider the case where an ordinary one-point Hermitian code is enlarged to an improved code. Otherwise, we would neglect certain cases where the order bound improved codes are in some sense `too good' to be used for enlargement as shown in the following example.
\begin{example}
  Consider the code pair $C_{\mathcal{L}}(D,52Q)\subsetneq C_{\mathcal{L}}(D,54Q)$ over $\mathbb{F}_{16}$. These codes have codimension $2$, and $C_{\mathcal{L}}(D,52Q)$ is self-orthogonal, meaning that we can apply Theorem~\ref{thm:qArySteane} to obtain a quantum code of dimension $2\cdot 47-64+2=32$ and minimal distance $d=\min\{12,(1+1/16)\cdot 10\}=11$. Using improved codes only, it is not possible to obtain as good parameters. The reason for this is that the codimension between $\tilde{E}(12)$ and $\tilde{E}(10)$ is only $1$. In fact, we have the inclusions
  \begin{equation*}
    C_{\mathcal{L}}(D,52Q)\subsetneq \tilde{E}(12)\subsetneq \tilde{E}(10)=C_{\mathcal{L}}(D,54Q).
  \end{equation*}
  Thus, if we restrict ourselves to improved codes only, we need to either start from a code smaller than $\tilde{E}(12)$ or enlarge to a code larger than $\tilde{E}(10)$. But neither option gives as good parameters as applying Steane-enlargement to $C_{\mathcal{L}}(D,52Q)\subsetneq\tilde{E}(10)$.
\end{example}
Despite the above observations, we shall refrain from stating the resulting parameters in a separate proposition since it would essentially say no more than Theorem~\ref{thm:qArySteane}. That is, such enlargements are in general not well-behaved enough to give meaningful formulae for their codimensions and minimal distances apart from the obvious ones, which already appear in Theorem~\ref{thm:qArySteane}.

To conclude this section, we give a few examples over $\mathbb{F}_{16}$ to illustrate the constructions presented in this section.
\begin{example}
  Let $q=4$, $\delta=20$, and consider $C_{\mathcal{L}}(D,(q^3-\delta)Q)=C_{\mathcal{L}}(D,44Q)$. As in Proposition~\ref{prop:enlargeUsualHermit} we set $\delta'=20-\lceil19/17\rceil=18$, and apply Theorem~\ref{thm:qArySteane} to the pair $C_{\mathcal{L}}(D,44Q)\subsetneq C_{\mathcal{L}}(D,46Q)$. This yields a quantum code with parameters $[[64,16,20]]_{16}$. Had we instead applied Corollary~\ref{cor:steaneSelfOrth} directly to $C_{\mathcal{L}}(D,44Q)$, the resulting parameters would be $[[64,14,20]]_{16}$.
\end{example}
\begin{example}
  The order bound improved code $\tilde{E}(5)$ is self-dual by Corollary~\ref{cor:selfOrthImprov}, and has parameters $[64,56,5]_{16}$. This code is contained in $\tilde{E}(4)$, which is a $[64,59,4]_{16}$-code. By applying the Steane-enlargement-technique, Theorem~\ref{thm:qArySteane}, we obtain a quantum code of length $64$, dimension $2\cdot 56-64+3$, and a minimal distance of at least
  \begin{equation*}
    \min\left\{5,\;\left\lceil\left(1+\frac{1}{16}\right)4\right\rceil\right\}=5.
  \end{equation*}
  That is, we can construct a $[[64,51,5]]_{16}$-quantum code. If only one-point Hermitian codes are used, the best quantum code of dimension $51$ has parameters $[[64,51,4]]_{16}$ stemming from $C_{\mathcal{L}}(D,60Q)\subsetneq C_{\mathcal{L}}(D,66Q)$.

  A graphical representation of the code inclusions can be found in Figure~\ref{fig:inclusions}.
\end{example}
\begin{figure}[hp]
  \centering
  \begin{tikzpicture}[xscale=0.7,yscale=0.5]
    \tikzset{every node/.style={inner sep=0,anchor=center,font=\scriptsize}}
    \path(-0.40,-0.40) -- (15.40,3.40); 
    \draw[draw=none,preaction={clip,postaction={fill=white, draw=Gray!25,line width=3pt, dotted}},rounded corners=2pt] (-0.40,-0.40) -- (-0.40,3.40) -- (12.40,3.40) -- (12.40,2.40)  -- (14.40,2.40) -- (14.40,0.40)  -- (15.40,0.40) -- (15.40,-0.40) -- cycle;
        
    \fill[fill=Gray!30,rounded corners=2pt] (-0.40,-0.40) -- (-0.40,3.40) -- (11.40,3.40) -- (11.40,2.40)  -- (13.40,2.40) -- (13.40,1.40)  -- (14.40,1.40) -- (14.40,-0.40) -- cycle;
    \fill[fill=Gray!70,rounded corners=2pt] (-0.40,-0.40) -- (-0.40,3.40) -- (0.40,3.40) -- (0.40,1.40)  -- (1.40,1.40) -- (1.40,0.40)  -- (3.40,0.40) -- (3.40,-0.40) -- cycle;
    \draw(0,0)node{0} (0,1)node{5} (0,2)node{10} (0,3)node{15} (1,0)node{4} (1,1)node{9} (1,2)node{14} (1,3)node{19} (2,0)node{8} (2,1)node{13} (2,2)node{18} (2,3)node{23} (3,0)node{12} (3,1)node{17} (3,2)node{22} (3,3)node{27} (4,0)node{16} (4,1)node{21} (4,2)node{26} (4,3)node{31} (5,0)node{20} (5,1)node{25} (5,2)node{30} (5,3)node{35} (6,0)node{24} (6,1)node{29} (6,2)node{34} (6,3)node{39} (7,0)node{28} (7,1)node{33} (7,2)node{38} (7,3)node{43} (8,0)node{32} (8,1)node{37} (8,2)node{42} (8,3)node{47} (9,0)node{36} (9,1)node{41} (9,2)node{46} (9,3)node{51} (10,0)node{40} (10,1)node{45} (10,2)node{50} (10,3)node{55} (11,0)node{44} (11,1)node{49} (11,2)node{54} (11,3)node{59} (12,0)node{48} (12,1)node{53} (12,2)node{58} (12,3)node{63} (13,0)node{52} (13,1)node{57} (13,2)node{62} (13,3)node{67} (14,0)node{56} (14,1)node{61} (14,2)node{66} (14,3)node{71} (15,0)node{60} (15,1)node{65} (15,2)node{70} (15,3)node{75} ;
  \end{tikzpicture}
  
  \medskip
  \begin{tikzpicture}[xscale=0.7,yscale=0.5]
    \tikzset{every node/.style={inner sep=0,anchor=center,font=\scriptsize}}
    \path(-0.40,-0.40) -- (15.40,3.40); 
    \draw[draw=none,preaction={clip,postaction={fill=white, draw=Gray!25,line width=3pt, dotted}},rounded corners=2pt] (-0.40,-0.40) -- (-0.40,3.40) -- (12.40,3.40) -- (12.40,2.40)  -- (14.40,2.40) -- (14.40,0.40)  -- (15.40,0.40) -- (15.40,-0.40) -- cycle;
    \fill[fill=Gray!30,rounded corners=2pt] (-0.40,-0.40) -- (-0.40,3.40) -- (11.40,3.40) -- (11.40,2.40)  -- (13.40,2.40) -- (13.40,1.40)  -- (14.40,1.40) -- (14.40,-0.40) -- cycle;
    \fill[fill=Gray!70,rounded corners=2pt] (-0.40,-0.40) -- (-0.40,3.40) -- (0.40,3.40) -- (0.40,1.40)  -- (1.40,1.40) -- (1.40,0.40)  -- (3.40,0.40) -- (3.40,-0.40) -- cycle;
    \draw(0,0)node{64} (0,1)node{59} (0,2)node{54} (0,3)node{49} (1,0)node{60} (1,1)node{55} (1,2)node{50} (1,3)node{45} (2,0)node{56} (2,1)node{51} (2,2)node{46} (2,3)node{41} (3,0)node{52} (3,1)node{47} (3,2)node{42} (3,3)node{37} (4,0)node{48} (4,1)node{43} (4,2)node{38} (4,3)node{33} (5,0)node{44} (5,1)node{39} (5,2)node{34} (5,3)node{29} (6,0)node{40} (6,1)node{35} (6,2)node{30} (6,3)node{25} (7,0)node{36} (7,1)node{31} (7,2)node{26} (7,3)node{21} (8,0)node{32} (8,1)node{27} (8,2)node{22} (8,3)node{17} (9,0)node{28} (9,1)node{23} (9,2)node{18} (9,3)node{13} (10,0)node{24} (10,1)node{19} (10,2)node{14} (10,3)node{9} (11,0)node{20} (11,1)node{15} (11,2)node{10} (11,3)node{5} (12,0)node{16} (12,1)node{12} (12,2)node{8} (12,3)node{4} (13,0)node{12} (13,1)node{9} (13,2)node{6} (13,3)node{3} (14,0)node{8} (14,1)node{6} (14,2)node{4} (14,3)node{2} (15,0)node{4} (15,1)node{3} (15,2)node{2} (15,3)node{1} ;
  \end{tikzpicture}
  
  \medskip
  \begin{tikzpicture}[xscale=0.7,yscale=0.5]
    \tikzset{every node/.style={inner sep=0,anchor=center,font=\scriptsize}}
    \path(-0.40,-0.40) -- (15.40,3.40); 
    \draw[draw=none,preaction={clip,postaction={fill=white, draw=Gray!25,line width=3pt, dotted}},rounded corners=2pt] (-0.40,-0.40) -- (-0.40,3.40) -- (12.40,3.40) -- (12.40,2.40)  -- (14.40,2.40) -- (14.40,0.40)  -- (15.40,0.40) -- (15.40,-0.40) -- cycle;
    \fill[fill=Gray!30,rounded corners=2pt] (-0.40,-0.40) -- (-0.40,3.40) -- (11.40,3.40) -- (11.40,2.40)  -- (13.40,2.40) -- (13.40,1.40)  -- (14.40,1.40) -- (14.40,-0.40) -- cycle;
    \fill[fill=Gray!70,rounded corners=2pt] (-0.40,-0.40) -- (-0.40,3.40) -- (0.40,3.40) -- (0.40,1.40)  -- (1.40,1.40) -- (1.40,0.40)  -- (3.40,0.40) -- (3.40,-0.40) -- cycle;
    \draw(0,0)node{1} (0,1)node{2} (0,2)node{3} (0,3)node{4} (1,0)node{2} (1,1)node{4} (1,2)node{6} (1,3)node{8} (2,0)node{3} (2,1)node{6} (2,2)node{9} (2,3)node{12} (3,0)node{4} (3,1)node{8} (3,2)node{12} (3,3)node{16} (4,0)node{5} (4,1)node{10} (4,2)node{15} (4,3)node{20} (5,0)node{9} (5,1)node{14} (5,2)node{19} (5,3)node{24} (6,0)node{13} (6,1)node{18} (6,2)node{23} (6,3)node{28} (7,0)node{17} (7,1)node{22} (7,2)node{27} (7,3)node{32} (8,0)node{21} (8,1)node{26} (8,2)node{31} (8,3)node{36} (9,0)node{25} (9,1)node{30} (9,2)node{35} (9,3)node{40} (10,0)node{29} (10,1)node{34} (10,2)node{39} (10,3)node{44} (11,0)node{33} (11,1)node{38} (11,2)node{43} (11,3)node{48} (12,0)node{37} (12,1)node{42} (12,2)node{47} (12,3)node{52} (13,0)node{41} (13,1)node{46} (13,2)node{51} (13,3)node{56} (14,0)node{45} (14,1)node{50} (14,2)node{55} (14,3)node{60} (15,0)node{49} (15,1)node{54} (15,2)node{59} (15,3)node{64} ;
  \end{tikzpicture}
  \caption{Graphical representation of the inclusions $\tilde{E}(5)^\perp\subsetneq\tilde{E}(5)\subsetneq\tilde{E}(4)$ over $\mathbb{F}_{16}$. The top grid shows $H^\ast(Q)$, the middle $\sigma(H^\ast(Q))$, and the bottom $\mu(H^\ast(Q))$. The darkly shaded region represents the elements whose evaluations make up $\tilde{E}(5)^\perp$. The lightly shaded region represent those that are also contained in $\tilde{E}(5)$, and the dashed line shows the additional elements in $\tilde{E}(4)$.
      Applying Steane-enlargement to the codes $\tilde{E}(5)\subsetneq\tilde{E}(4)$ yields a $[[64, 51, 5]]_{16}$ quantum code.}
  \label{fig:inclusions}
\end{figure}

\section{Comparison with existing constructions}\label{sec:comp-with-exist}
We will now compare the Steane-enlarged quantum codes from Section~\ref{sec:steaneHermitian} to some of those already in the literature. In order to conserve space, the examples presented in this section will primarily be those where the constructions of the present paper improve upon existing constructions. This is not meant to imply that such improvements can always be expected -- the cited works also contain specific examples of quantum codes whose parameters exceed what can be obtained using the results in Section~\ref{sec:steaneHermitian}.

For the order bound improved Hermitian codes from Section~\ref{sec:steaneHermitian}, we give in Table~\ref{tab:qchermit} a number of examples where the Steane-enlargement in Proposition~\ref{prop:hermitDimension} yields better parameters than those achievable in~\cite{qcHermit}. In all of these examples, the construction of \cite{qcHermit} gives an asymmetric quantum code where $d_z-d_x=1$. By using the Steane-enlargement technique, the minimal distance $d_x$ can be increased by one, yielding a symmetric quantum code of the same dimension. Had we not applied Steane-enlargement in these cases, we would have to resort to the lower of the minimal distances when considering symmetric codes.
\begin{table}[bp]
  \centering\scriptsize
  \begin{tabular}{cc}
    \toprule
    \textbf{Code} & \textbf{Dim. increase}\\
    \midrule
    $[[\ph\ph8, \ph\ph4, \ph3]]_{4\ph}$ & $2$\\
    $[[\ph27, \ph23, \ph3]]_{9\ph}$ & $2$\\
    $[[\ph27, \ph19, \ph4]]_{9\ph}$ & $2$\\
    $[[\ph27, \ph11, \ph7]]_{9\ph}$ & $2$\\
    $[[\ph64, \ph60, \ph3]]_{16}$ & $2$\\
    $[[\ph64, \ph56, \ph4]]_{16}$ & $2$\\
    $[[\ph64, \ph51, \ph5]]_{16}$ & $3$\\
    $[[\ph64, \ph40, \ph9]]_{16}$ & $2$\\
    $[[\ph64, \ph36, 10]]_{16}$ & $2$\\
    $[[\ph64, \ph30, 13]]_{16}$ & $2$\\
    $[[125, 121, \ph3]]_{25}$ & $2$\\
    $[[125, 117, \ph4]]_{25}$ & $2$\\
    $[[125, 112, \ph5]]_{25}$ & $3$\\
    $[[125, 107, \ph6]]_{25}$ & $2$\\
    $[[125, \ph97, \ph9]]_{25}$ & $2$\\
    $[[125, \ph91, 11]]_{25}$ & $2$\\
    $[[125, \ph79, 16]]_{25}$ & $2$\\
    $[[125, \ph75, 17]]_{25}$ & $2$\\
    \bottomrule
  \end{tabular}
  \hspace{0.5cm}
  \begin{tabular}{cc}
    \toprule
    \textbf{Code} & \textbf{Dim. increase}\\
    \midrule
    $[[125, \ph67, 21]]_{25}$ & $2$\\
    $[[343, 339, \ph3]]_{49}$ & $2$\\
    $[[343, 335, \ph4]]_{49}$ & $2$\\
    $[[343, 330, \ph5]]_{49}$ & $3$\\
    $[[343, 325, \ph6]]_{49}$ & $2$\\
    $[[343, 319, \ph7]]_{49}$ & $4$\\
    $[[343, 313, \ph8]]_{49}$ & $2$\\
    $[[343, 308, \ph9]]_{49}$ & $3$\\
    $[[343, 289, 15]]_{49}$ & $2$\\
    $[[343, 284, 16]]_{49}$ & $3$\\
    $[[343, 271, 21]]_{49}$ & $2$\\
    $[[343, 267, 22]]_{49}$ & $2$\\
    $[[343, 258, 25]]_{49}$ & $3$\\
    $[[343, 251, 29]]_{49}$ & $2$\\
    $[[343, 244, 31]]_{49}$ & $3$\\
    $[[343, 235, 36]]_{49}$ & $2$\\
    $[[343, 231, 37]]_{49}$ & $2$\\
    $[[343, 219, 43]]_{49}$ & $2$\\
    \bottomrule
  \end{tabular}
  \caption{Comparison between nearly symmetrical codes obtained via the procedure in Section~5 of \cite{qcHermit} and the Steane enlarged codes from this paper. A code $[[n,k,d]]_{q^2}$ listed here offers an improvement to a corresponding $[[n,k,d/(d-1)]]_{q^2}$ code from \cite{qcHermit}.
      All of the given codes retain their original minimal distance during enlargement, and the columns marked \emph{Dim. increase} indicate the increase in dimension when applying Theorem~\ref{thm:qArySteane} rather than Corollary~\ref{cor:steaneSelfOrth}.}
  \label{tab:qchermit}
\end{table}

To exemplify the advantage of using the order bound improved codes and the Steane-enlargement technique, Table~\ref{tab:techniqueComparison} shows a number of possible quantum code parameters over $\mathbb{F}_{16}$ when using different constructions based on the Hermitian curve. As is evident, the use of the order bound gives more knowledge on the minimal distance in column two, but also provides even better parameters when applying Steane-enlargement to the order bound improved codes.

\begin{table}[tp]
  \scriptsize
  \centering
  \begin{tabular}{cccc}
    \toprule
    \multicolumn{2}{c}{\textbf{One-point codes}} & \multicolumn{2}{c}{\textbf{Order bound improved}}\\
    \cmidrule(lr){1-2}\cmidrule(lr){3-4}
    \textbf{Goppa bound} & \textbf{Order bound} & \textbf{CSS} & \textbf{Steane-enlargement}\\
    \midrule
$[[64, 30, 12]]_{16}$   & $[[64, 30, 13]]_{16}$\imprMark   & $[[64, 30, 12]]_ {16}$\phMark    & $[[64, 30, 13]]_{16}$\phMark \\
$[[64, 32, 11]]_{16}$   & $[[64, 32, 11]]_{16}$\phMark     & $[[64, 32, 12]]_{16}$\imprMark   & $[[64, 32, 11]]_{16}$\phMark \\ 
$[[64, 34, 10]]_{16}$   & $[[64, 34, 10]]_{16}$\phMark     & $[[64, 34, 10]]_{16}$\phMark     & $[[64, 34, 10]]_{16}$\phMark \\ 
$[[64, 36, \ph9]]_{16}$ & $[[64, 36, \ph9]]_{16}$\phMark   & $[[64, 36, \ph9]]_ {16}$\phMark  & $[[64, 36, 10]]_{16}$\imprMark \\
$[[64, 38, \ph8]]_{16}$ & $[[64, 38, \ph9]]_{16}$\imprMark & $[[64, 38, \ph9]]_{16}$\phMark   & $[[64, 38, \ph9]]_{16}$\phMark \\
$[[64, 39, \ph7]]_{16}$ & $[[64, 39, \ph9]]_{16}$\imprMark & $[[64, 39, \ph6]]_{16}$\phMark   & $[[64, 39, \ph9]]_{16}$\phMark \\ 
$[[64, 40, \ph7]]_{16}$ & $[[64, 40, \ph8]]_{16}$\imprMark & $[[64, 40, \ph8]]_ {16}$\phMark  & $[[64, 40, \ph9]]_{16}$\imprMark  \\
$[[64, 42, \ph6]]_{16}$ & $[[64, 42, \ph6]]_{16}$\phMark   & $[[64, 42, \ph8]]_{16}$\imprMark & $[[64, 42, \ph7]]_{16}$\phMark \\
$[[64, 44, \ph5]]_{16}$ & $[[64, 44, \ph5]]_{16}$\phMark   & $[[64, 44, \ph6]]_{16}$\imprMark & $[[64, 44, \ph7]]_{16}$\imprMark  \\
$[[64, 45, \ph4]]_{16}$ & $[[64, 45, \ph5]]_{16}$\imprMark & $[[64, 45, \ph5]]_ {16}$\phMark  & $[[64, 45, \ph6]]_{16}$\imprMark  \\
$[[64, 46, \ph4]]_{16}$ & $[[64, 46, \ph5]]_{16}$\imprMark & $[[64, 46, \ph6]]_{16}$\imprMark & $[[64, 46, \ph5]]_{16}$\phMark \\
$[[64, 48, \ph3]]_{16}$ & $[[64, 48, \ph5]]_{16}$\imprMark & $[[64, 48, \ph5]]_{16}$\phMark   & $[[64, 48, \ph5]]_{16}$\phMark \\
$[[64, 50, \ph2]]_{16}$ & $[[64, 50, \ph4]]_{16}$\imprMark & $[[64, 50, \ph4]]_{16}$\phMark   & $[[64, 50, \ph5]]_{16}$\imprMark  \\
$[[64, 51, \ph0]]_{16}$ & $[[64, 51, \ph4]]_{16}$\imprMark & $[[64, 51, \ph4]]_{16}$\phMark   & $[[64, 51, \ph5]]_{16}$\imprMark  \\
$[[64, 54, \ph0]]_{16}$ & $[[64, 54, \ph4]]_{16}$\imprMark & $[[64, 54, \ph4]]_{16}$\phMark   & $[[64, 54, \ph3]]_{16}$\phMark \\
$[[64, 56, \ph0]]_{16}$ & $[[64, 56, \ph3]]_{16}$\imprMark & $[[64, 56, \ph3]]_{16}$\phMark   & $[[64, 56, \ph4]]_{16}$\imprMark  \\
$[[64, 58, \ph3]]_{16}$ & $[[64, 58, \ph3]]_{16}$\imprMark & $[[64, 58, \ph3]]_{16}$\phMark   & $[[64, 58, \ph3]]_{16}$\phMark \\
$[[64, 60, \ph0]]_{16}$ & $[[64, 60, \ph2]]_{16}$\imprMark & $[[64, 60, \ph2]]_ {16}$\phMark  & $[[64, 60, \ph3]]_{16}$\imprMark  \\
$[[64, 62, \ph0]]_{16}$ & $[[64, 62, \ph2]]_{16}$\imprMark & $[[64, 62, \ph2]]_{16}$\phMark   & $[[64, 62, \ph2]]_{16}$\phMark \\ 
    \bottomrule
  \end{tabular}        
  \caption{Comparison between different methods for constructing quantum codes from the Hermitian curve over $\mathbb{F}_{q^2}$.
      The codes in the first two columns stem from the CSS-construction applied to the usual one-point Hermitian codes, when bounding the distance by either the Goppa bound or the order bound.
      The third and fourth columns show the possible quantum code parameters when using order bound improved codes. In the third column, only the CSS-construction is used, and in the fourth Steane-enlargement is applied. Codes marked with {\imprMark} have better parameters than all preceding codes in the same row.}
  \label{tab:techniqueComparison}
\end{table}

\section*{References}\renewcommand{\chapter}[2]{\footnotesize}

\end{document}